\documentclass[runningheads,11pt]{llncs}

\newtheorem{observation}{Observation}

\usepackage{wrapfig}
\usepackage{authblk}

\usepackage{amsmath}
\usepackage{amssymb,graphicx}
\usepackage{algorithm}
\usepackage[noend]{algpseudocode}
\usepackage{float}
\usepackage{subcaption} %  for subfigures environments 
\usepackage{caption}

\usepackage{verbatim}
\usepackage{centernot,cancel}
\usepackage{xcolor}
\usepackage[colorinlistoftodos,prependcaption,textsize=tiny]{todonotes}
\usepackage{fancyhdr}
\usepackage{caption}
\usepackage{hyperref}

\usepackage{breakcites}

\usepackage{array}
\usepackage{xcolor,colortbl}

\newcolumntype{L}[1]{>{\raggedright\let\newline\\\arraybackslash\hspace{0pt}}m{#1}}
\newcolumntype{C}[1]{>{\centering\let\newline\\\arraybackslash\hspace{0pt}}m{#1}}
\newcolumntype{R}[1]{>{\raggedleft\let\newline\\\arraybackslash\hspace{0pt}}m{#1}}

\usepackage{makecell}

\newcommand*{\field}[1]{\mathbb{#1}}%

\newcommand{\Gr}{\mathcal{G}}

\def\lc{\left\lceil}   
\def\rc{\right\rceil}

\begin{document}

\title{Gathering in 1-Interval Connected Graphs
\thanks{All authors were supported by the EEE/CS initiative NeST. The last author was also supported by the Leverhulme Research Centre for Functional Materials Design. This work was partially supported by the EPSRC Grant EP/P02002X/1 on Algorithmic Aspects of Temporal Graphs.}}

\author{Othon Michail\inst{1}, Paul G. Spirakis\inst{1,2}, Michail Theofilatos\inst{1}}

\institute{Department of Computer Science, University of Liverpool, UK \and Computer Engineering and Informatics Department, University of Patras, Greece\\
	Email:\email{ \{Othon.Michail, P.Spirakis, Michail.Theofilatos\}@liverpool.ac.uk}}

\authorrunning{O. Michail, P. G. Spirakis, M. Theofilatos}
\titlerunning{Gathering in 1-Interval Connected Graphs}

\maketitle

\thispagestyle{empty}

\vspace{-0.8cm}

\begin{abstract} \normalsize
	We examine the problem of gathering $k \geq 2$ agents (or {\em multi-agent rendezvous}) in dynamic graphs which may change in every synchronous round but remain always connected ({\em $1$-interval connectivity}) \cite{KLO10}. The agents are identical and without explicit communication capabilities, and are initially positioned at different nodes of the graph.
	The problem is for the agents to gather at the same node, not fixed in advance.
	We first show that the problem becomes impossible to solve if the graph has a cycle. In light of this, we study a relaxed version of this problem, called {\em weak gathering}.
	We show that only in unicyclic graphs weak gathering is solvable, and we provide a deterministic algorithm for this problem that runs in polynomial number of rounds.
\end{abstract}

\vspace{-0.3cm}


\section{Introduction and Related Work} \label{sec:intro}

In \cite{LFPPSV18}, the authors study the feasibility of gathering $k \geq 2$ agents in $1$-interval connected rings and investigate the impact that {\em chirality} (i.e., common sense of orientation) and {\em cross detection} (i.e., the ability to detect whether some other agent is traversing the same edge in the same round) have on the solvability of the problem.
To enable feasibility, they empower the agents with some minimal form of implicit communication, called {\em homebases} (the nodes that the agents are initially placed are identified by an identical mark, visible to any agent passing by it).

%\subsection{Our contribution}

In this work we go beyond ring graphs and we start a characterization of the class of solvable $1$-interval connected graphs.
As we show, weak gathering is impossible if the graph contains at least two cycles, regardless of any other additional assumptions.

We then provide a deterministic algorithm that solves weak gathering in {\em unicyclic graphs}, and runs in polynomial number of synchronous rounds. A unicyclic graph is a connected graph containing exactly one cycle. Observe that ring graphs is a special case of unicyclic graphs.
The additional difficulty in these graphs comes from the fact that in most instances of initial agent configurations, the agents must gather on the cycle. However, in the model described in Section \ref{sec:model} the agents do not have the ability to distinguish the nodes that form the cycle. In Section \ref{section:weak-gathering-algorithm}, we empower the agents with some minimal form of implicit communication that allows them to assign identical labels on the nodes. We then carefully design a non-trivial mechanism that utilizes the graph topology and after $O(n^3\log{n})$ rounds the agents start moving only on the nodes of the cycle. Finally, the second part of the algorithm guarantees eventual correctness of {\em weak gathering}.
%as the communication model that we consider is very weak (the weakest possible that enables feasibility), the agents do not have the ability to assign labels on the nodes, thus recognize them when encountered again. This makes it impossible for the agents to identify the cycle, move there, agree on a meeting node and gather on that node.

%These problems model many situations that arise in the real world, e.g., searching for or regrouping animals, people, equipment, and vehicles.

\subsection{Model and Definitions}\label{sec:model}

\noindent \textbf{Static Network Model.}
A static network is modeled as an undirected connected graph $G_U = (V, E)$, referred to hereafter as a {\em static graph}.
The number of nodes $n = |V|$ of the graph is called its {\em size}.
Every node $u \in G_U$ has $\delta(u)$ incident edges, where $\delta(u)$ is its degree. For each of them, it associates a port and the ports are arbitrarily labeled with unique labels from the set $\{0, \dots, \delta(u) - 1\}$. We call these labels the {\em port numbers}.

\vspace{0.05cm}

\noindent \textbf{Dynamic Network Model.}
Given an underlying static graph on $n$ vertices, a {\em dynamic graph} on $G_U = (V, E)$ is a sequence $\mathcal{G_D} = \{G_t = (V, E_t) : t \in \field{N} \}$ of graphs such that $E_t \subseteq E$ for all $t \in \field{N}$. Every $G_t$ is the snapshot of $\mathcal{G_D}$ at time-step $t$. We assume that the sequence $\mathcal{G_D}$ is controlled by an adversarial scheduler, subject to the constraint that the resulting dynamic graph should be {\em $1$-interval connected}.
%An adversarial scheduler is said to be {\em unfair}, meaning that it can disable an edge $e_1$ at time $t' \geq 0$ forever (i.e., $e_1 \notin E_t, \; \forall t \geq t'$). %\in \field{N}$).

%In subsequent sections, we consider \textcolor{red}{worst-case and random schedulers}.
%A (worst-case) scheduler is said to be {\em unfair}, meaning that it can disable an edge $e_1$ forever ($e_1 \notin E_t, \; \forall t \in \field{N}$).
%while a {\em fair} scheduler activates all edges infinitely often during an execution.
%In addition, the class of dynamics that we consider is the $T$-interval-connectivity, that is, the system is fully synchronous, and the graph $G_D$ which consists of all nodes of $G$ and the union of all enabled edges during $T$ consecutive steps starting from any $t \geq 0$, should be connected.

\begin{definition}[$1$-interval-connectivity]
	A dynamic graph $G_D$ is $1$-interval-connected if for every integer $t \geq 0$, the static graph $G_t = (V, E_t)$ is connected.
\end{definition}

%\begin{definition}[$T$-interval-connectivity]
%	A dynamic graph $G_D$ is $T$-interval-connected, for an integer $T \geq 1$, if for every integer $i \geq 0$, the static graph $G_{[i,i+T]} = (V, \cap_{j=i}^{i+T-1}\mathcal{E}(j))$ is connected.
%\end{definition}

%\noindent In this work, we study exclusively the case where $T = 1$, i.e. $1$-interval connected graphs.

\noindent \textbf{Agents.}	The agents is a set $A = \{\alpha_1, \dots, \alpha_k\}$ of $k$ {\em anonymous} computational entities, each provided with memory and computational capabilities, that execute the same protocol.
They are arbitrarily placed on some nodes of the graph, and they are not aware of the other agents' positions.

More than one agent can be in the same node and may move through the same port number (i.e., the same edge) in the same round. We say that an agent $\alpha$ is {\em blocked} if the edge that $\alpha$ decided to cross in the current round is disabled by the scheduler.
%At the beginning of a round, all agents in a node $u$ can be in the waiting, incoming, or an outgoing buffer (if blocked in the previous round).
We consider the {\em strong multiplicity detection} model, in which each agent can count the number of agents at its current node. Based on that information, the port labeling and the contents of its memory, it determines whether or not to move, and through which port number.
When two or more agents move in opposite directions of the same edge in the same round, we can assume that they can either detect this event or not. If yes, we say that the system has {\em cross detection}.

We assume that the nodes of $G$ do not have unique identifiers, and the agents do not have {\em explicit} communication capabilities.
We do this in order to capture the limitations and the basic assumptions that make gathering in dynamic networks feasible.
Detailed assumptions needed by us to solve {\em weak gathering} in unicyclic graphs are clearly explained in Section \ref{section:weak-gathering-algorithm}.
% needed (i.e., impossibility of gathering), we either empower the agents with some form of communication, or/and cross detection.

\begin{definition}[Gathering problem]
	The {\em gathering} problem requires a set of $k$ mobile computational entities, called {\em agents}, initially located at different nodes of a graph, to gather within finite time at the same node, not known to them in advance.
\end{definition}

\begin{definition}[Weak gathering problem]
	The relaxed version of the gathering problem, called {\em weak gathering}, requires all agents to gather within finite time at the same node, or on the endpoints of the same edge. 
\end{definition}

\subsection{Impossibility results}

\begin{proposition}\label{lemma:gathering_1}
	{\em Gathering} is unsolvable in $1$-interval connected unicyclic graphs.
\end{proposition}
\begin{proof}
	Consider an underlying graph $G_U = (V, E)$, where there exists a single cycle $C$ of size $c > 3$. Assume that the number of agents is $k = 2$, and they try to solve gathering.
	$G_U$ can be represented as a ring graph, where each node $w \in C$ is the root of a tree $G_w$, $(C \setminus w) \notin G_w$, of size $s_w \geq 1$.
	Then, if the agents start from the same tree $G_w$, it is possible to meet without reaching the nodes of the ring. However, the agents are placed arbitrarily on the graph, thus, they might start from different trees. This means that the nodes must reach the ring in order to meet, in which case the scheduler can always block the path between them without violating the connectivity constraints.
	\qed
\end{proof}

\begin{proposition}\label{lemma:impossibility_1}
	{\em Weak gathering} is unsolvable in $1$-interval connected graphs, if $G$ has at least two cycles.
\end{proposition}
\begin{proof}
	Consider the case where $G$ contains exactly $2$ cycles $c_1$ and $c_2$ with no common vertices. All nodes of $c_1$ and $c_2$ can be roots of (independent) trees, and there is a single path connecting $c_1$ and $c_2$, through nodes $u \in c_1$ and $w \in c_2$. The nodes of this path can also be roots of (independent) trees.
	Consider now a partitioning of $G$ into $3$ groups $L$, $M$, and $R$, where $R$ and $L$ contain all the nodes of the cycles $c_1$ and $c_2$ respectively, and the trees starting from their nodes, except $u$ and $w$. $M$ contains all the nodes in the path between $u$ and $w$, including $u$ and $w$ and the trees starting from these nodes.
	
	The connectivity constraints imply that at most one edge from $c_1$ and one from $c_2$ can be missing in each round.
	If there are two agents that try to solve gathering and start from $R$ and $L$, the scheduler can always block the path to $u$ and $w$ without violating the connectivity constraints. This means that none of them can ever reach $u$ and $w$, thus, they can never meet or end up in neighboring nodes.
	
	Now, consider the case where $c_1$ and $c_2$ have at least two common vertices. Then, the partition $M$ contains all the common vertices (and the trees starting from these nodes), while $R$ and $L$ contain the rest of the nodes of $c_1$ and $c_2$ respectively (and the corresponding trees).
	The scheduler can again remove two edges from $G$ in each round from two different partitions (otherwise, if the scheduler removes two edges of the same partition, the connectivity constraints are violated). With a similar argument, the scheduler can always block an agent from reaching a different partition, thus, two agents that start from different partitions can never meet with each other, or move to neighboring nodes.
	
	In the case where $c_1$ and $c_2$ have only one common vertex $w$, if two agents start from $c_1 \backslash w$ and $c_2 \backslash w$, the scheduler can again block the path between the agents and $w$, by removing the corresponding edges in each cycle.
	
	Observe that the above cases apply also in graphs with more that two cycles. 
	This means that {\em weak gathering}, and the harder case of gathering cannot be solved in this setting.
	\qed
\end{proof}

%	\begin{observation}
%		Observe that the problem of {\em exploring} a graph inherits the impossibility results from gathering. This is because the (unfair) scheduler can completely block an agent from reaching some nodes of the $1$-interval connected unicyclic graph. In \cite{LDFS16} and \cite{MMM20} it is shown that further assumptions are required to make exploration solvable.
%	\end{observation}

Note that the above Propositions hold for any underlying graph with one and at least two cycles respectively.
The case of $1$-interval connected graphs without any cycle is equivalent to having a static tree graph, where the problem of agent gathering has been extensively studied.

\section{Weak gathering} \label{section:weak-gathering-algorithm}

In light of the above impossibilities, we hereafter consider unicyclic graphs, and we provide a deterministic algorithm that solves {\em weak gathering}.
The additional assumptions that we made are discussed and motivated later in this section. Briefly, we provide the agents with non-constant memory and knowledge of $k$.
In addition, for ease of presentation we also provide the agents with two identical movable tokens that can be placed on the nodes and picked up in subsequent visits of these nodes and are indistinguishable from the other agents' tokens, and we assume that the agents have cross-detection and knowledge of $n$.
In Section \ref{sec:drop_assumption} we discuss how we could drop the latter assumptions.

First, observe that in $1$-interval connected unicyclic graphs, the scheduler can block an agent from reaching some parts of the graph, however, all agents can move to the nodes of the unique cycle.
The scheduler is only allowed to remove one edge in each round, otherwise the connectivity constraints will be violated. Therefore, it is possible for the agents to reach two neighboring nodes and solve {\em weak gathering}.
In other words, if there is an execution of a {\em weak gathering} algorithm that an agent $\alpha$ never reaches the cycle, then there is a sequence of (connected) graphs where $\alpha$ never reaches the same or neighboring node with the rest of the agents.
% They must be all at the same time in the cycle...

\begin{observation}\label{observation:weak-gathering-1}
	In order to achieve {\em weak gathering} in $1$-interval connected unicyclic graphs, the agents must gather on the nodes of the cycle.
\end{observation}

\noindent Call $C$ the set of nodes of the unique cycle and $G_w$ the (connected) tree starting from node $w \in C$ and $(C \setminus w) \notin G_w$. The above observation holds because an agent in a node $v \in G_w$ can be completely blocked from reaching any other node $u \notin G_w$. This means that {\em weak gathering} can only be achieved if the agents first reach a node in $C$ (otherwise they will not be on the same or neighboring nodes).
The above observation means that the agents must first perform some sort of exploration on the graph in order to reach the cycle and gather on some node $v \in C$.

\vspace{0.1cm}

\noindent \textbf{Communication assumptions:}
A very common assumption that makes the problem solvable in ring graphs is for the agents to have distinct identities \cite{CPL12,MDKKPV06,DFP06}. Alternatively, another assumption which pertains to the communication capabilities of the agents, is either to supply each node with a {\em whiteboard} where the agents can leave notes as they travel \cite{SKSOKM20}, mark the nodes that the agents are initially placed (identifiable and identical nodes called {\em homebases}), or provide the agents with a constant number of movable tokens that can be placed on nodes, picked up, and carried while moving \cite{CDKK08}. Under the first communication assumption the problem becomes solvable even in the presence of some faults \cite{BFFS07,CJSS07}.
In \cite{LFPPSV18}, the authors used homebases in order to break the symmetry in $1$-interval connected ring graphs.
In our work, we choose to empower the agents with two movable and identical tokens (called hereafter \textit{pebbles}). The pebbles of each agent are also indistinguishable from those of another agent. In Section \ref{sec:drop_assumption} we discuss how this assumption could be dropped.

%The following property follows by slightly generalizing Property $2$ of \cite{LFPPSV18} (i.e., in unicyclic graphs the agents must gather on the cycle).

%\begin{proposition}
%	{\em Weak gathering} in unicyclic graphs is unsolvable if the homebases are not distinguishable from the other nodes.
%\end{proposition}

%In our setting the agents can be initially placed anywhere on the graph, meaning that a homebase might not exist in any node of the unique cycle.
%We argue that even with distinguishable homebases there are configurations where the problem becomes unsolvable.
%Therefore, we choose to empower the agents with two movable and identical tokens (called hereafter \textit{pebbles}). The pebbles of each agent are also indistinguishable from those of another agent.

%For simplicity of presentation we choose to empower the agents with two movable and identical tokens (called hereafter \textit{pebbles}), \textcolor{red}{and we later explain how to drop this assumption, and instead use homebases to obtain the same results.}

\vspace{0.1cm}

\noindent \textbf{Cross detection:}
In the algorithm of Di Luna et al. \cite{LFPPSV18} the authors considered the case where cross detection is available, and then they construct a mechanism which avoids agents crossing each other (i.e., no agents traverse the same edge at the same round and in opposite directions), called {\em Logic Ring}, in order to drop this assumption.
We now assume that the agents have cross-detection, and in Section \ref{sec:drop_assumption} we explain how this mechanism can be applied in our setting.

\vspace{0.1cm}

\noindent \textbf{Memory requirements:}
A very significant aspect of mobile agent systems is the memory requirements of the agents. In \cite{budach1978}, the authors show that the problem of exploring a static graph with a finite state automaton (or agent) is unsolvable if the port numbers of the nodes are set arbitrarily. In \cite{FIPPP05}, the authors show that $\Theta(D\log{d})$ bits of memory are required to achieve exploration, where $D$ is the diameter, and $d$ the maximum degree of the graph.
An alternative approach to network design consists in graph preprocessing by setting the port numbers, so that graph exploration is easy, i.e., with constant memory \cite{GKMNZ07, ilcinkas2008}.
In this work, we provide the agents with non-constant memory, as we assume that the port labels are assigned arbitrarily.

\vspace{0.1cm}

\noindent \textbf{Knowledge of $n$ and $k$:}
Finally, we assume that the agents know the size $n$ of the graph and the number $k$ of agents.
We first show that if $k$ is not known, then {\em weak gathering} is unsolvable. Finally, in Section \ref{sec:drop_assumption} we discuss how to drop the assumption of knowing $n$.

\begin{property}
	If $k$ is not known, then {\em weak gathering} is unsolvable, even with homebases.
\end{property}
\begin{proof}
	A well known result even for static graphs is that if neither $n$ nor $k$ are known, then gathering is unsolvable, regardless of chirality and cross detection.
	We now show that if $k$ is not known, then gathering is unsolvable in unicyclic graphs.
	
	By Observation \ref{observation:weak-gathering-1}, weak gathering can only be achieved on the nodes of the cycle, otherwise the scheduler can choose two agents and completely block them from reaching the same node or the endpoints of the same edge.
	
	Consider the case where only $n$ is known. The agents in order to decide that gathering was achieved must find in some way the number of agents. This is because if they manage to reach the same node $u$, they need to decide whether to terminate, wait there, or continue moving on the graph. Assume that $u$ is on the cycle. 
	Observe that the case where they terminate, or wait on $u$ in order for the rest of the agents to gather leads to impossibility as the scheduler can choose one agent that has not reached there and block it from reaching $u$ or any neighboring node of $u$. Assume now that all agents have reached $u$. If they choose to move on the graph (even with a stronger communication model which allows them to communicate and move on the graph as a group), this would happen indefinitely.
	
	Assume now that homebases are distinguishable from the rest of the nodes. Because of the above, the agents need in some way to infer the number of agents on the graph. They are not allowed to communicate in any way, thus, each agent needs to explore the graph, visit and count all homebases.
	However, because of the cycle and the fact that the agents cannot leave labels on the nodes of the graph, every time an agent traverses the cycle, all nodes are seen as unexplored, and there is no way of determining whether a node was explored in a previous step or not (i.e., a homebase was previously visited). This means that a single homebase might be counted more than once, thus, the agents will eventually not hold the correct value of $k$.
	\qed
\end{proof}

%\begin{property}
%	If $n$ is not known, then gathering is unsolvable.
%\end{property}
%\begin{proof}
%	Maybe...
%\end{proof}

%\subsection{Feasible configurations.} 
%...

\subsection{Weak gathering algorithm}

Our deterministic algorithm is divided into two phases, and the overall idea is the following:
During the first phase all agents place one of their pebbles on their initial nodes.
Then, they start exploring the graph using a DFS approach.
Each agent $\alpha$ gradually moves its pebble closer to the cycle, and when its pebble reaches the cycle, $\alpha$ moves to the second phase.
When all agents have moved to the second phase, the executed process ensures that they will eventually gather on the cycle.
%Note that the agents might move to the second phase before all pebbles reach the cycle. In this case, gathering might be unsuccessful; the agents recognize that and they move back to the first phase again.

In order to make the description of the algorithm more clear, we first introduce a number of variables that are stored in the local memory of each agent.

\begin{itemize}
	\item \textit{round}: Counter that is increased by one in each round.
	\item \textit{Graph (or $\Gr$)}: Contains the lists that represent the nodes visited by an agent. A specific node of the underlying graph might correspond to multiple nodes in $\Gr$. We refer to the {\em Graph} of an agent $\alpha$ as $\Gr_\alpha$.
	\item \textit{stepsAway}: The distance between the agent and its pebble in $\Gr$.
	\item \textit{epoch}: The epoch, which determines the maximum distance of the DFS exploration in $\Gr$ ($2^{epoch}$).
	\item \textit{roundsBlocked}: The number of rounds that the agent is continuously blocked.
	\item \textit{pebblesFound}: The number of pebbles found in distinct nodes of $\Gr$.
\end{itemize}

\noindent \textbf{Phase 1.} This phase is responsible for traversing the graph (exploration) and identifying the nodes that form the cycle.
We now present all procedures that take place during this phase.
%\begin{enumerate}
%	\item Traversing the graph (exploration).
%	\item Moving the pebbles on the nodes that form the cycle.
%	\item Find the nodes that form the cycle, and recognize that {\em weak gathering} can be achieved.
%\end{enumerate}

\vspace{0.1cm}

\noindent {\em Mapping of the graph.}
The problem of graph mapping has been extensively studied in the literature. Most of the algorithms rely on either the usage of {\em whiteboards} \cite{DFNS05,DFKNS07}, or assume that the agents can observe the memory contents of each other when they meet on the same node \cite{GTKC12}. In the latter, the agents maintain multiple hypotheses when ambiguity about the graph topology occurs, and they resolve it when they meet.

Clearly, in order to explore and map an anonymous graph, the agents need to mark the nodes, so as to identify previously visited nodes on subsequent visits. However, in our work the marks (i.e., pebbles) made by each agent are indistinguishable from those made by another, thus, it is not clear whether multiple agents can successfully map an anonymous graph.
Despite of the above problem, we have carefully designed an algorithm that correctly maps part of the graph and allows the agents to successfully identify the cycle without having to resolve the ambiguity that occurs.

Each agent $\alpha$ stores in its local memory a list of the neighbors of each vertex visited and the port numbers that led to those nodes. Let $u$ be the initial node of an agent $\alpha$. Then, $\alpha$ constructs a list $L(u)$ which represents $u$. Assume that it traverses an edge through port number $i$ and arrives at a node $w$ at port number $j$. It then constructs a new list $L(w)$, and in $L(u)$ stores $i$ and a pointer to $L(w)$. At the same time, it stores in $L(w)$ the port $j$ and a pointer to $L(u)$.
We call these lists the {\em Graph} of $\alpha$ or $\Gr_\alpha$.

\vspace{0.1cm}

\noindent {\em Graph exploration.}
We use a traditional technique which makes each agent traverse a tree in the DFS way. In a round, when the agent arrives at node $u$ through a port $i$, it leaves $u$ through port $(i + 1)$ mod $\delta(u)$ in the next round (if the edge is available). Initially, the agents start by leaving the port $0$.
We divide the execution into epochs, and in each epoch $e$ the agents perform DFS up to distance $2^e$ in $\Gr$. In particular, when {\em stepsAway} $= 2^e$, the agent moves through the port that it arrived from.
Initially $e = 0$, and when the agent returns to its initial node and has traversed all its neighbors in that phase, it increases $e$ by one.

\vspace{0.1cm}

\noindent {\em Pebbles on the cycle.} During the exploration process, the agents check a number of predicates that help them to gradually move their pebble to the cycle and achieve gathering.
In particular, whenever an agent $\alpha$ reaches a leaf, it marks the node in $\Gr_\alpha$ with a special character $\cancel{C}$, indicating that it does not belong to the cycle, and never moves to that node (of $\Gr_\alpha$) again. In addition, if a node $u \in \Gr_\alpha$ with degree $\delta(u)$ has $\delta(u) - 1$ marked neighbors, the agent also marks $u$.
Whenever an agent $\alpha$ reaches the node $u \in \Gr_\alpha$ containing its own pebble, and $u$ is marked with $\cancel{C}$, $\alpha$ moves its pebble to the (unique) neighboring node $w \in \Gr_\alpha$ that is not marked.

\vspace{0.1cm}

\noindent {\em Cycle detection.}
When an agent $\alpha$ encounters a pebble, it marks the node in $\Gr_\alpha$ (locally) with a special character $\mathcal{T}$, and increases the value of the counter \textit{pebblesFound} by one (if already marked with $\mathcal{T}$ or $\cancel{\mathcal{T}}$, it does nothing). In particular, only the first time that it visits a node with a pebble during an epoch $e$ it marks the corresponding node in $\Gr_{\alpha}$. This is important in order to maintain a consistent knowledge about the positions of the pebbles as we show in Lemma \ref{lemma:weak-gatherin-3}.
We hereafter call the $\mathcal{T}$-marked nodes of $\Gr$ {\em pebbles of $\Gr$} or $\mathcal{T}_{\Gr_\alpha}$, and the set of nodes with pebbles on the underlying graph {\em distinct pebbles} or $P$.
We say that a node $w \in \mathcal{T}_{\Gr_\alpha}$ corresponds to a node $u \in P$, and we write $w \rightarrow u$, if $w$ was marked when the agent $\alpha$ visited $u$. If $\forall u \in P$, $\exists w \in \mathcal{T}_{\Gr_\alpha}: w \rightarrow u$, we say that the agent $\alpha$ visited all distinct pebbles, and we write $\mathcal{T}_{\Gr_\alpha} \equiv P$. 
When $\alpha$ counts $k + 1$ pebbles ({\em including its own pebble}), we show that it can verify whether these pebbles are on the cycle or not. In case that they are on the cycle, the agent enters to the second phase of the algorithm which only moves on the nodes of the cycle and eventually achieves {\em weak gathering}.
At this point, $\mathcal{T}_{\Gr_\alpha}$ contains the nodes of $\Gr_\alpha$ in which $k + 1$ pebbles were found by an agent $\alpha$.
Given $\mathcal{T}_{\Gr_\alpha}$ and $\Gr_\alpha$, agent $\alpha$ constructs in its local memory a (shortest) path $\mathcal{P} = \{e_0, e_1, \dots, e_s\}$ with vertex sequence $\mathcal{V} = \{v_0, v_1, \dots, v_s\}$, from the first to the last node of $\mathcal{T}_{\Gr_\alpha}$ in $\Gr_\alpha$, containing all nodes of $\mathcal{T}_{\Gr_\alpha}$.
Here, $e_l = (p_i, p_j)$, where $p_i$ is the port number that led to $v_l$, and $p_j$ the port number of $v_l$ that it arrived at (i.e., the port numbers in the endpoints of the same edge).
If $\mathcal{P}$ is not a line in $\Gr_\alpha$, it marks all nodes of $\mathcal{T}_{\Gr_\alpha}$ with a different character $\cancel{\mathcal{T}}$ and resets \textit{pebblesFound} to zero (this can be achieved in one round locally, without visiting all nodes of $\mathcal{T}_{\Gr_\alpha}$ again).
%Then, it moves back to the node where it was in the beginning of the {\em cycle detection} step and continues the exploration of the graph.
Otherwise, it constructs a cycle $C'$ assuming that the first and last pebbles $f, \; l$ visited correspond to the same pebble $p \in P$, i.e., $f \rightarrow p$ and $l \rightarrow p$.
In other words, it assumes that the nodes $v_0, \; v_s \in \mathcal{V}$ correspond to the same node of the underlying graph.
To construct $C'$, it sets $\mathcal{P} = \{e_1, \dots, e_s\}$ and when it reaches $v_s$ it moves through the port specified by $e_1$.
Then, it traverses $C'$ in order to verify if it is a cycle (the ports visited and the locations of the pebbles must agree with $C'$). If yes (note that this can only happen if its own pebble is on the actual cycle of the underlying graph), it enters to the second phase of the algorithm.
Otherwise it moves back to the node where it was in the beginning of the {\em cycle detection} step by following the reverse path, and continues with the exploration of the graph.
In order to avoid situations where an agent $\alpha$ is blocked in consecutive unsuccessful {\em cycle detections}, it does not count again the pebbles that belong on the same nodes of $\Gr$ until the end of its current epoch (i.e., it marks all nodes in $\mathcal{T}_{\mathcal{G}_a}$ with $\cancel{\mathcal{T}}$).
Note that if $C$ is the cycle of the underlying graph, then $C'$ might be a set of multiple traversals of $C$. For example, if $C = \{v_0, v_1, \dots, v_m\}$, then $C' = \{v_0, v_1, \dots, v_m, v_0, v_1, \dots, v_m, v_0, v_1, \dots\}$.
If $C \equiv C'$, the agents solve {\em weak gathering}.
Otherwise, {\em weak gathering} will fail and the agents continue moving only on the nodes of $C'$ (second phase) until they find again $k + 1$ pebbles, in which case, they reconstruct $C'$.
%In addition, this subroutine of the algorithm is only executed after the $5n$ round has passed (i.e., $\textit{round} \geq 5n$).
%Note that due to the cycle, some pebbles can be counted more than once, while the agent is not aware of this fact. We handle this in the second phase of the algorithm, and we explain how the agents achieve gathering.
Finally, note that $\alpha$ can now traverse the cycle both clockwise and counterclockwise, though, the orientation might be different for each agent. We later explain how to obtain chirality (i.e., common sense of orientation).

\vspace{0.2cm}

%\noindent \textit{Blocked agents. } Due to the $1$-interval connectivity, the scheduler can only remove edges on the cycle, and only one edge can be missing in each round. We handle the case where an agent is blocked during the first phase of the algorithm as follows.
%At the beginning of an epoch, the agents set \textit{roundsBlocked} to zero, and in each round, if an agent is blocked, it increases \textit{roundsBlocked} by one and waits there until either the edge becomes available again (in this case it resets \textit{roundsBlocked} to zero), or the termination condition is satisfied.
%When $\textit{roundsBlocked} \geq 2^e - \textit{stepsAway}$ at any time during the first phase, the agent starts moving towards its pebble. If the agent successfully reaches its pebble, it continues with the exploration of the graph by entering to the next phase. If not (i.e., the scheduler blocks the path of the agent to its pebble), it continues to increase \textit{roundsBlocked} by one in each round.
%In the {\em Termination condition} we describe how we handle the case where an agent is blocked for very long time.

\noindent \textbf{Phase 2.} When an agent $\alpha$ enters to this phase, it means that is has constructed a cycle $C'$ in its local memory, and all nodes of $C'$ are on the actual cycle of the underlying graph.
In this phase the agents assume that all pebbles have reached the cycle, and they perform some actions that would solve weak gathering in case that this assumption is true.
An agent can either be in state  {\em walking} or {\em gathering}, and initially it is in state {\em gathering}.
In {\em Grouping} we explain how the agents form groups when certain predicates are satisfied.
We call a set of agents a {\em group} if they are on the same node and move in the same direction.
In the first state ({\em walking}), it traverses the cycle counterclockwise (according to its own sense of orientation), and when it visits $k + 1$ pebbles (i.e., $\text{pebblesFound} = k + 1$), it reconstructs a cycle $C'$ as explained in {\em cycle detection} and changes its state to {\em gathering}.
In the second state, based on the distances between the pebbles in $\Gr_\alpha$ and the port labeling in cycle $C'$ it elects a node $u \in C'$ as the meeting point. In {\em Unique node election} we explain how this is achieved.
If for any reason the agents do not agree on $u$, {\em weak gathering} will not succeed, thus they reset \textit{pebblesFound} to zero and their state becomes \textit{walking}.
As we show later, after $O(n^3\log{n})$ rounds all pebbles reach the cycle, and phase $2$ will eventually succeed.
If they agree on $u$, they can also obtain chirality by utilizing the port numbers of $u$.
Consider a node $u \in C$ with ports $p_1$ and $p_2$ that lead to its neighboring nodes in $C'$. Assume, w.l.o.g., that $p_1 < p_2$. Then, $\alpha$ sets as {\em clockwise} the orientation that is defined by traversing $p_1$ and {\em counterclockwise} the one defined by traversing $p_2$.

After determining the node $u$ where they should meet, the agents move for $2n$ rounds towards $u$ by following the shortest path (we call this the {\em first step} of state {\em gathering}).
We now distinguish the following cases for an agent $\alpha$, depending on whether $\alpha$ reached $u$, or not, after $2n$ rounds.
%\begin{itemize}
%	\item Agent $\alpha$ reached the $u$.
%	\item Agent $\alpha$ did not reach $u$ after $2n$ rounds (due to missing edges).
%\end{itemize}
If an agent arrived at node $u$ after $2n$ rounds in state {\em gathering}, it checks whether all agents are there. If yes, it terminates. Otherwise, it starts moving clockwise on the cycle for $n$ rounds ({\em second step} of state {\em gathering}).
As we show later, by the end of round $2n$ all agents that entered to state {\em gathering} during a time window of length $n$ are divided into at most two groups.
The rest of the agents that due to missing edges did not reach the elected node after $2n$ rounds in state {\em gathering}, they start moving counterclockwise as a group for $n$ rounds ({\em second step}).
We want the agents in each group to start the {\em second step} of this phase at the same time (the two groups may start at different rounds). However, observe that the agents might not enter to state {\em gathering} at the same time. In {\em Grouping}, we explain how the agents start walking on the cycle as groups.
At this point, there are two groups of agents moving towards each other.
In any case, the two groups of agents will either end up on the same node, or they will cross each other, or they will become blocked on the endpoints of the same edge.
In {\em Grouping}, we explain how these groups of agents merge after at most $n$ rounds or terminate in neighboring nodes.

\vspace{0.1cm}

\noindent {\em Blocked agents and termination condition.} The overall idea is that if an agent is blocked long enough for the rest of the agents to reach some endpoint of the missing edge, then weak gathering is achieved and the agents terminate.
To achieve this, in each round, if an agent $\alpha$ is blocked, it increases \textit{roundsBlocked} by one and waits there until either the edge becomes available again (in this case it resets \textit{roundsBlocked} to zero), or the termination condition is satisfied. In particular, if $\textit{roundsBlocked}_{\alpha} \geq \delta n\log{n}$, for some small constant number $\delta$ and {\em no other agent arrived at the same node during that round}, it places its second pebble and terminates. The rest of the agents on the same node recognize that the number of pebbles on that node was increased by one, thus they terminate.
There are two cases for the rest of the agents that are on the other endpoint $w$ of the missing edge. Either the missing edge becomes enabled again, or it remains disabled long enough for some other agent $\alpha' \in w$ to place its second pebble on $w$, and the same process occurs. In the first case, all agents move on the other endpoint where the rest of the agents are (and have terminated), they count that the total number of agents is $k$, and they terminate.

\vspace{0.1cm}

\noindent {\em Unique node election.} The goal of this subroutine of the algorithm is to break the symmetry in the cycle $C$ of the graph and elect a unique node as the meeting point of the agents.
This is feasible in most cases, given the topology of the underlying graph $G$, the port labeling of the nodes, and the distances between the final positions of the pebbles in $C$.
It is very important that all agents elect, eventually, the same node, but this can be achieved only if the information used to break the symmetry between the agents is identical. 
During the exploration phase, an agent $\alpha$ stores in its local memory information about the nodes as it visits them, in $\Gr_\alpha$.
However, due to the cycle and because the agents are placed arbitrarily on the graph, the graphs $\Gr$ of two agents might be different.
%In addition, the scheduler can completely block an agent from reaching some nodes of the graph, thus, never acquiring complete information about the graph.
This means that $\Gr_\alpha$ cannot be used to break the symmetry, otherwise they might elect different nodes.
However, during the second phase of the algorithm, all agents construct in their local memories a cycle $C'$ which will eventually be identical to each other. Thus, the final positions of the pebbles and the port labeling of the nodes in $C'$ can lead to the election of a unique node as their meeting point.
If the above configuration is symmetric, the agents recognize that {\em weak gathering} cannot be achieved. In this case, they continue executing the algorithm, as the pebbles might have not reached their final positions.

\vspace{0.3cm}

\noindent {\em Grouping.} This subroutine of the algorithm is used in order to form groups of agents in the following cases. 

(1) During the {\em first step} in state {\em gathering}, the agents move towards the elected node for $2n$ rounds. However, not all agents start this step at the same time. The first predicate of {\em grouping} is responsible to synchronize the agents so as to begin the {\em second step} at the same time, and then continue moving as a group.
In particular, when an agent $\alpha$ counts $2n$ rounds of phase $2$, it then moves either clockwise or counterclockwise, depending on whether it reached the elected node or not.
Let $u$ be the node where $\alpha$ was at the end of the {\em first step} of phase $2$. Then, it moves only for one round and waits there (at most $n$ rounds) for the rest of the agents in $u$ to move at the same node as $\alpha$. To achieve this, the agents in $u$ detect that the number of agents was decreased by one. If $\alpha$ crossed an agent during that round, it moves back to $u$ and repeats the same process again. When the rest of the agents in $u$ reach $\alpha$, they continue moving as a group.

(2) If an agent in state {\em walking} visits the elected node $u$ and there are some other agents there, it assumes that they are in state {\em gathering}. In this case, it enters to state {\em gathering} and waits there at most $2n$ rounds, or until the first predicate of {\em grouping} is satisfied.

%(3) Consider a group of agents walking counterclockwise, and some agents are in state {\em walking} while the rest of them are in state {\em gathering}. If at some round the size of the group becomes smaller, the agents in state {\em walking} reverse direction for $n$ rounds.

(3) When two agents or groups of agents cross each other or visit the same node, they merge into a single group. To achieve this, when this happens, the group which is closer to the elected node $u$ (say $G_1$) by following the clockwise path, reverses direction. The other group $G_2$ waits until $G_1$ catches them. Then, the agents that were in state {\em gathering} continue walking in their initial direction, while the agents in state {\em walking} reverse direction (if not already did).
After a successful edge traversal of $G_1$, if $G_2$ is missing, it reverses direction again.
%(observe that this is the case where two agents or groups of agents in state {\em walking} crossed each other).
In order to avoid situations where the two agents or groups of agents get stuck, after their next successful edge traversal they do not try to group (for one round).
Similarly, if the groups of agents visit the same node, the agents in state {\em walking} reverse direction and perform the same procedure as in the previous case. Here, the group of agents in state {\em gathering} does not do anything.
Finally, after a successful merging, the agents in state {\em walking} change to state {\em gathering}.
In the cases where the edge between the two groups is missing, they wait until it becomes available again, or until the {\em termination condition} is satisfied.

%(3) If a group of agents in state {\em gathering} during the second step does not cross any other agent (or group), or does not meet any other agents in a node, after $n$ rounds it reverses direction and moves for $n$ more rounds. In this way we guarantee that if there is a group of agents in state {\em walking} and a group in state {\em gathering} (that did not reach the elected node), after $n$ steps they will traverse the cycle in opposite directions in order to merge into a single group.

If in any of the previous cases the number of agents is $k$, they all terminate.

%Otherwise (unsuccessful gathering), they reset {\em pebblesFound} to zero and continue moving counterclockwise in $C'$ (state {\em walking}).
%In the second case, the group of agents which is closer to $u$ by following the counterclockwise orientation (i.e., the agents that did not reach $u$) waits there, while the rest of the agents reverse direction. Then, they will either end up on the same node, where they count the number of agents and perform the same actions as in the first case, or the edge between them will be missing.
%In the cases where the edge between the two groups is missing, they wait there until the edge will become available again, or until the {\em termination condition} is satisfied. In the first case, the groups of agents will cross each other, and the same process occurs.

\begin{algorithm}[H]
%	\SetAlgoLined
%	\KwResult{Identifies the nodes that form the cycle.}
	\textbf{Result:} Identifies the nodes that form the cycle. 

	(1) Initialization of variables.\;
	
	(2) Place a pebble on the initial node.\;
	
	(3) Explore graph up to distance $2^e$, and create a mapping in $\Gr$. Mark all nodes of $\Gr$ where pebbles are found with $\mathcal{T}$.\;
	
	(4) Mark all nodes with $\cancel{C}$ in $\Gr$ that have exactly one unmarked neighbor. When you mark the node where the pebble is, move it on the unique unmarked neighbor.\;
	
	(5) Upon marking $k+1$ nodes with $\mathcal{T}$, construct the (shortest) path $C'$ that contains all these nodes and traverse it. If successfully traversed, move to {\em Phase $2$}.
	Otherwise, mark these nodes with $\cancel{\mathcal{T}}$, move to initial node and continue with the exploration. 
	
	\caption{First phase}
\end{algorithm}

\begin{algorithm}[H]
%	\SetAlgoLined
%	\KwResult{Solves {\em weak gathering} on some node of the cycle.}
	\textbf{Result:} Solves {\em weak gathering} on some node of the cycle.
	
	(1) State {\em walking}:
	
	(i) Move counterclockwise.
	
	(ii) Upon counting $k+1$ pebbles, elect a leader node $u$ using the port labeling and the positions of the pebbles, and change to state {\em gathering}.
	
	\vspace{0.2cm}
	
	(2) State {\em gathering}:
	
	(i) Move towards $u$ for $2n$ rounds.
	
	(ii) If reached $u$, after round $2n$ move clockwise for $n$ rounds. Otherwise, move counterclockwise for $n$ rounds.
	
	\vspace{0.2cm}
		
	(3) {\em Grouping} and {\em termination}:

	(i) In each round check all predicates of {\em grouping} and perform the corresponding actions.
	
	(ii) If at any time the number of agents is $k$, or the number of rounds that it is blocked is more than $\delta n \log{n}$, terminate.
	
%	(i) If reached $u$ (in any state), and there are some other agents there, change to step (i) of {\em gathering}.
	
%	(ii) If crossed some other agent(s) and is closer to $u$ by following the clockwise path, reverse direction. Otherwise, wait. Then, if it is in state {\em walking}, reverse direction (if not already did), otherwise (state {\em gathering}) continue walking towards the initial direction.

	\caption{Second phase}
\end{algorithm}

\subsection{Analysis}

We first show that after the end of the first phase of the algorithm, all agents correctly identify the nodes that form the cycle $C$, and then they only move on $C$.
In addition, because of the fact that an agent can be blocked on a node of $C$ indefinitely, we show that during the first phase all agents reach some endpoint of the missing edge after at most $\delta n\log{n}$, for some small constant number $\delta$, rounds.

We then continue and show that in the second phase of the algorithm all agents eventually enter to state {\em gathering}, and they correctly solve {\em weak gathering}.

\subsubsection{First phase of the algorithm}

\begin{lemma}\label{lemma:weak-gathering-1}
	Let $d_t(p_\alpha, C)$ denote the (shortest) distance between the pebble $p_\alpha$ of an agent $\alpha$ and the closest node of the cycle $C$ at round $t$. Then, $\{d_t(p_\alpha, C)\}$, $t \geq 0$ is a decreasing sequence (i.e., $d_t \geq d_{t+1}$).
\end{lemma}
\begin{proof}
	Initially, the agents are arbitrarily placed on some nodes of the graph. During the first phase they place their pebbles on these nodes and start the exploration of the graph in epochs $e$ in a DFS way, and up to a maximum depth which depends on the epoch ($\textit{stepsAway} = 2^e$).
	
	Call $C$ the unique cycle and $G_u$ the (connected) tree starting from node $u \in C$ and $(C \setminus u) \notin G_u$, where an agent $\alpha$ is initially placed.
	As agent $\alpha$ moves on the graph, it constructs in its local memory the graph $\Gr_{\alpha}$.
	In order to mark a node $w \in \Gr_{\alpha}$ with $\cancel{C}$, all its neighbors $v$ except one must already be marked in $\Gr_\alpha$. This can only happen initially on the leaf nodes, then their neighbors, and so on.
	Now observe that all the nodes of the cycle (including $u$) have two neighbors that belong to the cycle $C$, thus, $\alpha$ cannot mark any of them.
	This means that all the nodes in the shortest path between the current position of its pebble and $u$ are not marked in $\Gr_{\alpha}$, while the rest of the nodes $v \in G_u$ will eventually be marked.
	When $\alpha$ marks the node that its pebble is, it picks it, and moves it on the unique neighbor that is not marked. Similarly, the above argument will be satisfied for the new position of its pebble.
	Because of this fact, a pebble can only move closer to the cycle every time the corresponding agent moves it, and eventually it will reach $u$.
	\qed
\end{proof}

\begin{lemma}[Cycle detection]\label{lemma:weak-gathering-7}
	When an agent enters phase $2$, its pebble is on the cycle, and it only moves on the nodes of the cycle.
\end{lemma}
\begin{proof}
	When an agent $\alpha$ encounters a pebble on an unmarked node of $\Gr_\alpha$, it marks it with a special character $\mathcal{T}$, and when it marks $k + 1$ nodes it executes the {\em cycle detection} procedure.
	Call $C$ the set of nodes that form the cycle in the underlying graph and $\mathcal{T}_{\Gr_\alpha}$ the set of nodes that an agent $\alpha$ marked in $\Gr_\alpha$.
	
	The {\em first step} of the {\em cycle detection} procedure is to check whether the shortest path $\mathcal{P} = \{e_0, e_1, \dots, e_s\}$, with vertex sequence $\mathcal{V} = \{v_0, v_1, \dots, v_s\}$ which connects all the nodes in $\mathcal{T}_{\Gr_\alpha}$ is a line in $\Gr_\alpha$, or not.
	If not, it resets its \textit{pebblesFound} variable to zero, changes their marks to $\cancel{\mathcal{T}}$, and the procedure stops.
	Otherwise, it constructs a cycle $C'$ in its local memory and traverses it in order to verify if all nodes of $C'$ are on the cycle.

	%In the second step, observe that the nodes of $\mathcal{V}$ correspond to a path on the underlying graph $G$ as follows:
	%$\mathcal{V}' = \{u_1, u_2, \dots, u_i, c_1, c_2, \dots, c_l, w_1, w_2, \dots, w_j\}$, where $u_{[i]}, w_{[j]} \notin C$, $i,j \geq 0$, and the $c_{[l]}$ nodes correspond to one or more traversals of $C$. Otherwise, $\mathcal{V}$ is not a line in $\Gr_\alpha$.
	%Assume that the last element of $\mathcal{V}$ is the current position of the agent $\alpha$. We now distinguish three cases depending on the values of $i$ and $j$.
	
	%(1) $i \geq 0$ and $j > 0$: In this case, the current node of the agent $\alpha$ is not on the cycle $C$. This means that at some node $v$, it will need to traverse an edge through which it arrived at $v$, or it will not meet the rest of the pebbles (at least one is on $C$). In both cases, the {\em cycle detection} fails.
	
	%(2) $i > 0$ and $j = 0$: In this case, the current position of the agent is on $C$, and there is at least one pebble that is not on $C$. Assume that the pebble $p_1$ is the furthest pebble of $C$....
	
	%(3) $i = j = 0$: In this case, all pebbles visited by $\alpha$ belong on the cycle. Note that $C'$ might correspond to multiple traversals of $C$, meaning that more than one distinct pebble was visited more than once. In any case, the cycle detections will succeed and the agent will enter to the second phase of the algorithm which only moves on $C'$ (i.e., $C$).

	In this step, while traversing $C'$, if the port numbers do not match with the ones visited, the {\em cycle detection} fails.
	Observe that if there exists a node $w \in \mathcal{V}$ such that $w \notin C$, then this procedure will fail. This is because it will either need to traverse the edge through which it arrived at a node that is not on the cycle, or it will not meet the $k+1$ pebbles of  $\mathcal{T}_{\Gr_\alpha}$. This holds even in the case where the rest of the agents moved their pebbles. Then, the agent $\alpha$ will surely pass through the nodes of the initial positions of the pebbles and the procedure will fail.
	
	If {\em cycle detection} succeeds, then all nodes of $C'$ are on the cycle (including the node where its own pebble is); $\alpha$ enters to the second phase and it only moves on these nodes.
	\qed
\end{proof}

%Observe that an agent in order to count $k + 1$ pebbles it needs to traverse the whole cycle at least once. This is because even if it visits all nodes of the underlying graph, it can count at most $k$ pebbles. In order to expand the graph $\Gr$ in its local memory, thus visit the same nodes of the underlying graph without being able to distinguish them, it needs to traverse the cycle first.

In contrast to the literature on exploration of graphs, in our model the agents cannot assign distinct labels on the nodes, thus recognize them when encountered again (cf., e.g., \cite{PP1998}).
This difficulty comes from the fact that the communication model that we consider does not allow the agents to write information on the nodes other than leaving a constant number of identical {\em pebbles} (i.e., one bit of information), {\em indistinguishable} for all agents.
For this reason, when an agent enters to the cycle and completes a tour, the whole graph is again considered as unexplored.
However, in our algorithm we guarantee that after $O(n^3\log{n})$ rounds, all pebbles reach the cycle and all agents enter to the second phase which solves {\em weak gathering}.

\begin{lemma}\label{lemma:weak-gathering-2}
	The number of rounds until all pebbles reach the cycle is bounded by $O(n^3\log{n})$.
\end{lemma}
\begin{proof}
	Call $C$ the unique cycle and $G_u$ the (connected) tree starting from node $u \in C$ and $(C \setminus u) \notin G_u$, where an agent $\alpha$ is initially placed.
	Let $w$ be the initial node of $\alpha$, and $d(u, w) > 0$ the (shortest) distance between $u$ and $w$.
	In order for an agent to move its pebble (by Lemma \ref{lemma:weak-gathering-1} closer to the cycle), it must first explore all nodes of $G_u$ in the worst case (i.e., reach all the leaves of $G_u$).
	The agents start from epoch $e = 0$ and perform DFS up to distance $2^e$. In the worst case, the diameter of $G_u$ is $|G_u| - 1$. When $2^e \geq |G_u - 1| \Rightarrow e = \lc \log{(|G_u| - 1)} \rc$ the agent moves its pebble on $u$ by the end if this epoch.
	
	The number of steps in each epoch depends on the topology of the graph. In particular, when an agent enters the cycle and completes a tour, the whole graph can again be considered as unexplored, thus the agent continues exploring nodes that has already visited in previous rounds.
	
	In an epoch $e$, $\frac{2^e}{|C|}$ complete tours can occur, where $|C|$ is the size of the cycle.
	The total number of complete tours of the cycle until epoch $e = \lc \log{(|G_u| - 1)} \rc$ can be bounded by:
	\begin{equation}
	T = \sum_{e = 0}^{\lc \log{(|G_u - 1|)} \rc} \frac{2^e}{|C|} = \frac{2^{\lc \log{(|G_u| - 1)} \rc + 1} - 1}{|C|} < 2 \frac{|G_u|}{|C|}
	\end{equation}
	
	Due to the $1$-interval connectivity, the scheduler can block $\alpha$ when it wants to traverse an edge of the cycle.
	During the DFS exploration, the number of edge traversals on the cycle is $2|C|$ for every complete tour of it. Now observe that if $\alpha$ is blocked for more than $ \delta n\log{n}$, for some small constant number $\delta$, rounds, it terminates, and as we show later gathering is achieved. This means, that in the worst case which does not lead to gathering, the scheduler blocks the agents for $\delta n\log{n} - 1$ rounds for each edge traversal in $C$.
	In addition, for each cycle tour, $n - |C|$ nodes (in the worst case) can be explored without being blocked by the scheduler.
	
	In addition, when agent $\alpha$ visits $k + 1$ pebbles (including its own pebble), it executes the {\em cycle detection} subroutine of the algorithm which performs $|C'|$ edge traversals, where $C'$ is the cycle constructed in the local memory of $\alpha$.
	We have that $e \leq \lc \log{(|G_u| - 1)} \rc \leq \lc \log{n} \rc$, thus $|C'| \leq \epsilon n$, for some small constant number $\epsilon$.
	
	By Lemma \ref{lemma:weak-gathering-7}, for this to succeed, its own pebble first needs to reach the cycle. However, the distance between its pebble and the cycle is $d(u, w) > 0$, thus, it will fail. Then, $\alpha$ performs $|C'|$ more edge traversals in order to reach the node where it was in the beginning of the {\em cycle detection} step.
	For each edge traversal, as before, $\alpha$ can remain blocked for $\delta n\log{n} - 1$ rounds. For each complete tour of the cycle, it can perform {\em cycle detection} only once. This is because after an unsuccessful cycle detection, the agents mark with $\cancel{\mathcal{T}}$ the nodes with pebbles in $\Gr_\alpha$.
	
	Therefore, the total number of rounds needed for a pebble to reach the cycle, considering the worst case of the scheduler choices can be bounded by:
	\begin{equation}
	\begin{split}
	S &= T ((2|C| + 2\epsilon n)(\delta n\log{n} - 1) + 2(n-|C|)) \\ &= O(n^3\log{n})
	\end{split}
	\end{equation}
	\qed
\end{proof}

\begin{lemma}\label{lemma:weak-gatherin-3}
	Each agent in phase $1$ of the algorithm visits all nodes of the cycle every $O(n\log{n})$ rounds, if not blocked by the scheduler.
	%Let $\alpha$ be an agent. During the first $O(n^3)$ rounds, or until $\alpha$ enters phase $2$, it traverses the cycle every at most $5n$ rounds after leaving it.
\end{lemma}
\begin{proof}
	Consider an agent $\alpha$, initially in node $u$ and distance $d_1 = d(u, C)$ from the cycle $C$.
	When $2^e = d_1 + |C|/2 \Rightarrow e = \log{(d_1 + |C|/2)} \leq \log{n}$, the agent traverses all nodes of the cycle for the first time.
	The number of rounds in each epoch depends on the topology of the graph.
	Let $S_e$ be the set of nodes which are in distance at most $2^e$ from the initial position of agent $\alpha$, and $S_e$ does not contain all nodes of $C$ (if it contains all nodes of $C$, then it will traverse $C$ prior to the $\log{n}$ -th phase in the worst case).
	Then, the total number of rounds of the DFS exploration until epoch $e = \log{n}$ is $\sum_{e=0}^{\log{n}}2S_e \leq \sum_{e=0}^{\log{n}}2n = n\log{n}$. This means that it takes $n\log{n}$ rounds to reach all nodes of the cycle for the first time, and then $e \geq \log{n}$ guarantees that in each phase the agent traverses the cycle every at most $2n$ rounds (DFS exploration) in phase $1$.
	
	%In the worst case, the total number of rounds is  rounds to reach the cycle and $n$ more to traverse it, thus $5n$ in total.
	%During the first $5n$ rounds, the agents do not mark the nodes where pebbles are found, thus, the {\em cycle detection} subroutine can be executed only after the $5n$-th round.
	
	%A factor that can delay an agent from reaching (and traversing) the cycle is the {\em cycle detection} subroutine.
	%By Lemma \ref{lemma:weak-gathering-7}, an agent $\alpha$ enters to the second phase of the algorithm only after its pebble has reached the cycle, and then, it only moves on the nodes of $C$.

	In order to find the number of rounds between the cycle traversals we need to study the number of times that the {\em cycle detection} procedure is executed, which may delay an agent from visiting all nodes of $C$.
	Observe that only if the shortest path between the $\mathcal{T}$-marked nodes is a line in $\Gr_\alpha$ the agent stops the DFS exploration and traverses $C'$.
	We now distinguish two cases.
	
	(1) $C'$ contains all nodes of $C$: The agent $\alpha$ marks all $k$ nodes with pebbles and performs a complete tour of the cycle in order to visit a $(k+1)$-th pebble. In this case, {\em cycle detection} might be executed, however, $C'$ contains all nodes of $C$. This means that when the {\em cycle detection} is executed, the agent visits all nodes of the cycle after at most $n$ rounds. In case that this procedure fails, the agent traverses again the cycle after at most $n$ rounds, and then continues with the DFS exploration. Otherwise, the agent enters to the second phase of the algorithm, which by Lemma \ref{lemma:weak-gathering-7} moves only on $C$.

	(2) $C'$ does not contain all nodes of $C$: Let $T_{\Gr_\alpha}$ be the $\mathcal{T}$-marked nodes of $\alpha$ in $\Gr_{\alpha}$, and $|T_{\Gr_\alpha}| = k + 1$. Let $P$ be the set of nodes with pebbles of the underlying graph, and $T_{\Gr_\alpha} \cancel{\equiv} P$ (there is some node in $P$ which was not visited by $\alpha$).
	This can be the result of some other agent $\alpha'$ moving its pebble $p$ on a different node and $\alpha$ has marked both positions of $p$ (otherwise $C'$ would contain all nodes of $C$).
	However, in order for $\alpha$ to start the {\em cycle detection} procedure, all nodes in $T_{\Gr_\alpha}$ have to be on a line in $\Gr_\alpha$. This means that even in the case where $\alpha'$ moved its pebble, $\alpha$ will either pass through the node where $p$ initially was, in which case $\alpha$ decreases \textit{pebblesFound} by one and unmarks the corresponding node in $\Gr_\alpha$, or it will not mark the new position of $p$, as this will not be the first time that it visited that node during that phase.
	In both cases, $\alpha$ will not execute the {\em cycle detection} procedure.
	
	%No complete tours by default here.
	%If pebble in different tree, on its way forward sees it in u, and on its way back, and because the pebbles only move towards the cycle, it has to pass through u.
	%If same tree, and distance from cycle is bigger, again, on its way back it will move through the initial position of the pebble.
	%IF distance is smaller, 
	
	%two cases, (1) the pebble moves closer to the agent. In this case, it will move through the initial position.
	%(2) the pebble moves further from its own pebble. In this case and as it only marks a node in $\Gr_{\alpha}$ only the first time it visits it during a phase, it will not mark it on the way back, thus, it will not start the {\em cycle detection} procedure.

	Finally, for $e \leq \log{n}$, the agents traverse the cycle after at most $n\log{n}$ rounds. For $e \geq \log{n}$, they traverse it every at most $3n$ rounds.
	\qed
\end{proof}

\subsubsection{Second phase of the algorithm}

We now show that phase $2$ of the algorithm successfully gathers all agents either at the same node, or at the endpoints of the same edge.
%We first show that after the first step of the algorithm in state {\em gathering}, where the agents move for $2n$ rounds towards the elected node, all agents that have entered phase $2$ are divided into two groups by the end of the $2n$ round.	

\begin{lemma}\label{lemma:weak-gathering-termination}
	Let the variable $\textit{roundsBlocked}_{\alpha}$ of an agent $\alpha$ be $\delta n\log{n}$, for some small constant number $\delta$. Then, all agents are gathered on the endpoints of the missing edge and terminate.
\end{lemma}
\begin{proof}
	Let $\alpha$ be an agent that is blocked on some node $u$ of the cycle. By Lemma \ref{lemma:weak-gatherin-3} and because of the fact that the scheduler can only remove at most one edge in each round, all other agents in phase $1$ of the algorithm perform a block-free execution, thus, after at most $\delta n\log{n}$ rounds, for some small constant number $\delta$, they traverse the cycle and they reach $u$. An agent in phase $2$ of the algorithm can either be in state {\em walking} or {\em gathering}. In the first case, after at most $n$ rounds, it reaches $u$. In the second case, an agent needs $2n$ rounds to move towards the elected meeting point and then it walks the cycle (either clockwise or counterclockwise) for $n$ more rounds, which is enough to reach $u$.
	
	Finally, all $k$ agents end up on the two endpoints of the missing edge, the {\em termination condition} of $\alpha$ and then of the rest of the agents, is eventually satisfied.
	\qed
\end{proof}

\begin{lemma}\label{lemma:weak-gathering-8}
	Consider a set of agents $S$ moving towards a node $u$ in cycle $C$, following the shortest path. After $n$ rounds the agents of $S$ are in at most two nodes of $C$, and one of them is in $u$.
\end{lemma}
\begin{proof}
	Consider a set of agents $S_1 \in S$ moving clockwise and a set of agents $S_2 \in S$ moving counterclockwise.
	Consider two agents $\alpha_1, \; \alpha_2  \in S_1$ moving towards $u$.
	Assume that in the shortest path to $u$, the distance between $\alpha_1$ and $u$ is $d_1$ and the distance between $\alpha_2$ and $u$ is $d_2$.
	
	The number of successful edge traversals until they reach $u$ is at most $n/2$.
	Assume that $\alpha_1$ didn't reach $u$ after $n$ rounds. This means that it was blocked for at least $n/2 + 1$ rounds. Since $1$-interval connectivity in this setting allows only one edge to be missing in each round, $\alpha_2$ can be blocked for at most $n/2 - 1$ rounds (when not in the same node with $\alpha_1$). Thus, if $d_1 < d_2$, $\alpha_2$ reaches $\alpha_1$ by round $n$, and if $d_1 > d_2$, it reaches $u$ by round $n$.
	
	Now consider an agent $\alpha_3 \in S_2$ moving towards $u$ (different orientation from $\alpha_1$ and $\alpha_2$).
	Since $\alpha_1$ was blocked for at least $n/2 + 1$ rounds and the agents follow the shortest path to $u$ (they cannot be blocked on the endpoints of the same edge), $\alpha_3$ can be blocked for at most $n/2 - 1$ rounds. Thus it reaches $u$ by round $n$.
	
	Overall, if an agent $\alpha$ is blocked for more than $n/2 + 1$ rounds, then all agents that move in the same orientation towards $\alpha$ reach $\alpha$ by round $n$, while the rest of the agents reach $u$. Otherwise, all agents reach $u$ by round $n$.
	\qed
\end{proof}

\iffalse

\begin{lemma}\label{lemma:weak-gathering-5}
	Consider a set of agents $S$, where $|S| < k$, that are in phase $2$ of the algorithm in state {\em gathering}. Then, gathering either fails and all agents in $S$ enter to state \textit{walking} after at most $\delta n\log{n} + 2n$ rounds, for some small constant number $\delta$, or weak gathering is achieved and they terminate.
\end{lemma}
\begin{proof}
	Independently of the choices of the scheduler, in the first step of phase $2$, by Lemma \ref{lemma:weak-gathering-8}, after $n$ rounds all agents of $S$ are divided into two groups. 
	Then, for the next $n$ steps they move towards each other; they end up either on the same node, or the two groups of agents cross each other, or the two groups reach the endpoints of the same (missing) edge.
	In the first case they verify that $|S| < k$ and they enter to state \textit{walking}. In the second case, the group which is closer to the elected (meeting) node via the clockwise path reverses direction and the other group waits there.
	In the case where the edge between them is missing, they wait at most $\delta n\log{n}$ rounds or until the edge will become available again.
	During these rounds, and as the scheduler can only remove one edge in each round, by Lemma \ref{lemma:weak-gatherin-3} all agents traverse the cycle after at most $\delta n\log{n}$ rounds without being blocked. If blocked, it will be on one endpoint of the same missing edge and weak gathering will be achieved.
\end{proof}
\fi

\begin{theorem}\label{theorem:weak-gathering-6}
	All agents after $O(n^3\log{n})$ rounds enter phase $2$, elect the same node as the meeting point and solve weak gathering.
\end{theorem}
\begin{proof}
	By Lemma \ref{lemma:weak-gathering-2}, after $O(n^3\log{n})$ rounds all pebbles reach the cycle, and when this happens, by Lemma \ref{lemma:weak-gathering-7}, all agents move only on the nodes of the cycle (in phase $2$).
	
	Let $r'$ be the round that the last agent traverses all nodes of the cycle for the first time in the second phase. This means that after $r'$ and because all agents have obtained the same information (i.e., the port labeling and the locations of the pebbles), all agents agree on the meeting node.
	Let also $R = \{r_1, r_2, \dots, r_k\}$ be the rounds that the $k$ agents enter to state {\em gathering} for the first time after $r'$.
	%Observe that $|r_i - r_j| < n, \; r_i, r_j \geq r' \text{and} \forall i,j$.
	
	In the second phase of the algorithm the agents can either be in state {\em walking} or {\em gathering}.
	Consider a set of agents $S_1$ that are in state {\em gathering} and $|r_i - r_j| < n, \; \forall i,j \in S_1$, and a second set of agents $S_2$ contains the rest of the agents.

	At this point, by Lemma \ref{lemma:weak-gathering-8} all agents that enter phase $2$ at the same time, after at most $n$ rounds are divided into at most two groups $G_1$ and $G_2$, and one of them (say $G_1$) is on the elected node $u$. Thus, after $2n$ rounds all agents of $S_1$ are divided into two groups. In addition, the {\em Grouping} subroutine guarantees that the agents of $G_1$ and $G_2$ will continue moving as groups during the second step of phase $2$.
	We now consider two cases.
	
	(1) All agents of $S_1$ reached $u$. Then, some of the agents of $S_2$ reach $u$ and enter to state {\em gathering}, and the rest of the agents of $S_2$, again by Lemma \ref{lemma:weak-gathering-8}, they become a group that did not reach $u$ due to missing edges. Observe that this group walks the cycle counterclockwise, while the agents of $S_1$ walk the cycle clockwise. At this point there are two groups of agents moving towards each other. Therefore, the {\em Grouping} procedure guarantees that after at most $n$ rounds the two groups will either merge (in this case they terminate), or they will become blocked on the endpoints of the same edge until the {\em termination condition} will be satisfied.

	(2) In this case, the agents of $S_1$ are divided into two groups at round $r_1 + 2n$. During the first $2n$ rounds, some of the agents of $S_2$ may reach $u$, thus enter to state {\em gathering} and continue moving as a group with $G_1$.
	
	(a) If the agents of $G_2$ move clockwise, then the rest of the agents of $S_2$ may cross the agents of $G_2$ or arrive on the same node. In both cases they will merge into a single group.
	
	(b) If the agents of $G_2$ move counterclockwise, then the rest of the agents of $S_2$ end up on the same node with the agents of $G_2$. This is because the agents of $G_2$ remain blocked long enough that at round $r_1 + 2n$ they did not reach $u$. Then, all of the agents in the clockwise path between $G_2$ and $u$ after $2n$ rounds reach $G_2$.
	In this case, the agents of $S_2$ will reverse direction (to clockwise), however $G_2$ will continue moving counterclockwise. Then, they reverse direction again because the agents of $G_2$ move counterclockwise. This procedure continues until they will either reach $u$, or until some agent in $G_2$ enters to the second step of phase $2$. Then, they will cross each other and {\em Grouping} guarantees that they will merge into a single group.

	%Finally, at this point there are two groups of agents moving towards each other. The {\em grouping} procedure then guarantees that after at most $n$ rounds they all merge
	
	In all these cases, all agents reach either the same node and the {\em termination condition} is satisfied, or they become blocked at the endpoints of the same missing edge where, by Lemma \ref{lemma:weak-gathering-termination} they solve {\em weak gathering}.
	\qed
\end{proof}

\subsection{Towards dropping the additional assumptions}\label{sec:drop_assumption}

\subsubsection{Knowledge of $n$}

Observe that in our algorithm, $n$ is used in two cases.
The first case is on the {\em termination condition} where the agents terminate if they are blocked long enough for the rest of the agents to reach the same node or the other endpoint of the same (missing) edge.
If we assume that the agents do not know $n$ then it is not clear how and whether it is possible to achieve termination.

Our algorithm also uses $n$ during the second phase, where the agents need $n$ in order to guarantee complete tours of the cycle $C$. In this case we can replace $n$ with $|C'|$, where $C'$ is the locally constructed cycle of an agent. This is because for all agents $|C'| \geq |C|$ throughout the execution.
%This means that Lemma \ref{lemma:weak-gathering-8} holds also for $|C'|$

\subsubsection{Cross detection}

In \cite{LFPPSV18} the authors provide a mechanism which avoids agent crossing. In particular, each agent constructs an edge labeled bidirectional ring, such that the intersection of the labels assigned in the edges of the clockwise direction with the ones of the counterclockwise direction is empty. Then, the agents move on the actual ring subject to the constraint that at round $r$ they can traverse an edge only if the set of labels of that edge contains $r$.
This guarantees that two agents moving in opposite directions will never cross each other on an edge of the actual ring.

The above construction works only if all agents have the same reference point and have obtained the same sense of orientation.
In the second phase of our algorithm, all agents after $O(n^2 \log{n})$ rounds either solve {\em weak gathering} (by being blocked for very long), or traverse all nodes of the cycle and elect the same node as their meeting point.
Therefore, the labels of the {\em logic rings} of all agents eventually become the same and by slightly modifying our algorithm (e.g., allow $4n$ rounds during the {\em first step} in state {\em gathering}, and $2n$ during the {\em second step}), the Theorem \ref{theorem:weak-gathering-6} follows.

\subsubsection{Pebbles}

In our algorithm each agent is supplied with two identical pebbles. It uses one pebble in the {\em cycle detection} subroutine, and the second one in the {\em termination condition} in order to notify the rest of the agents that it was blocked for more than $\delta n \log{n}$ rounds.

For both cases, it might be possible to substitute the pebbles with knowledge of $n$. In particular, when an agent $\alpha$ is blocked for $\delta n \log{n}$ rounds it terminates. Then, the rest of the agents are gathered in the endpoints of the same edge, and they move towards each other. After a successful edge traversal of the group where the agent $\alpha$ terminated, they count that the number of agents was decreased by one. In this case they reverse direction. The other group of agents should then wait for the first group to catch them up.

Regarding the detection of the cycle, when an agent explores a path $P$ of size more than $n$, it means that the cycle $C$ is a subpath of $P$. When an agent finds such a path of size more than $n$, it can move to the first node of $P$ and start traversing all subpaths of size $i, \; 3 \leq i \leq n$ for at least $n$ rounds each. If the agent performs $n$ successful edge traversals and the port labeling of the nodes visited matches the one of $P$, then we believe that the agent has successfully identified the nodes that form the cycle.

\section{Open problems}

An immediate open problem is whether we can achieve the same results if the class of dynamics is the $T$-interval connectivity, for $T > 1$. If we consider probabilistic algorithms, can we find a more efficient algorithm for {\em weak gathering} in unicyclic graphs?
In addition, can we extend the class of solvable graphs if we impose a fairness assumption to the scheduler?
A very interesting question is whether the second phase of our algorithm can be replaced by a modified asynchronous version of the algorithm of Di Luna et al. (\cite{LFPPSV18}), where the starting times of the agents might be different.

%In our work we showed feasibility of {\em weak gathering} in $1$-interval connected unicyclic graphs, by constructing an algorithm that runs in $O(n^3\log{n})$ time. Are there any more time efficient algorithms for this setting?
In Section \ref{sec:drop_assumption} we argued that the communication model that we considered (i.e., the pebbles) can be substituted by knowledge of $n$.
Finally, in this setting it is not clear how to achieve termination without empowering the agents with knowledge of $n$.
We gave an intuition of how to achieve both, however we leave them as open problems.

%Finally, in the setting where the agents can obtain chirality, in \cite{LFPPSV18} the authors show that knowledge of either $k$ or $n$ is required in order to enable feasibility. In case that either $k$ or $n$ is known, it utilizes the network topology (ring graph) in order to discover the other. However, in our setting, it is not clear how to achieve similar results. In our algorithm, $n$ is used only in the second phase of the algorithm where it can maybe be replaced by $|C'|$. In addition, the {\em cycle detection} procedure, makes use of the number of agents and it is not clear how to achieve the same without knowledge of it. We leave these questions as open problems.

%\newpage

\bibliographystyle{alpha-abr}
{\normalsize \bibliography{main}}

\newcommand{\etalchar}[1]{$^{#1}$}
\begin{thebibliography}{DMGK{\etalchar{+}}06}

\bibitem[BFFS07]{BFFS07}
L.~Barriere, P.~Flocchini, P.~Fraigniaud, and N.~Santoro.
\newblock Rendezvous and election of mobile agents: impact of sense of
  direction.
\newblock {\em Theory of Computing Systems}, 40[2]:143--162, 2007.

\bibitem[Bud78]{budach1978}
L.~Budach.
\newblock Automata and labyrinths.
\newblock {\em Mathematische Nachrichten}, 86[1]:195--282, 1978.

\bibitem[CDKK08]{CDKK08}
J.~Czyzowicz, S.~Dobrev, E.~Kranakis, and D.~Krizanc.
\newblock The power of tokens: rendezvous and symmetry detection for two mobile
  agents in a ring.
\newblock In {\em International Conference on Current Trends in Theory and
  Practice of Computer Science}, pages 234--246. Springer, 2008.

\bibitem[CDS07]{CJSS07}
J.~Chalopin, S.~Das, and N.~Santoro.
\newblock Rendezvous of mobile agents in unknown graphs with faulty links.
\newblock In {\em International Symposium on Distributed Computing}, pages
  108--122. Springer, 2007.

\bibitem[CPL12]{CPL12}
J.~Czyzowicz, A.~Pelc, and A.~Labourel.
\newblock How to meet asynchronously (almost) everywhere.
\newblock {\em ACM Transactions on Algorithms (TALG)}, 8[4]:1--14, 2012.

\bibitem[DFK{\etalchar{+}}07]{DFKNS07}
S.~Das, P.~Flocchini, S.~Kutten, A.~Nayak, and N.~Santoro.
\newblock Map construction of unknown graphs by multiple agents.
\newblock {\em Theoretical Computer Science}, 385[1-3]:34--48, 2007.

\bibitem[DFNS05]{DFNS05}
S.~Das, P.~Flocchini, A.~Nayak, and N.~Santoro.
\newblock Distributed exploration of an unknown graph.
\newblock In {\em International Colloquium on Structural Information and
  Communication Complexity}, pages 99--114. Springer, 2005.

\bibitem[DFP03]{DFP06}
A.~Dessmark, P.~Fraigniaud, and A.~Pelc.
\newblock Deterministic rendezvous in graphs.
\newblock In {\em European Symposium on Algorithms}, pages 184--195. Springer,
  2003.

\bibitem[DLFP{\etalchar{+}}18]{LFPPSV18}
G.~A. Di~Luna, P.~Flocchini, L.~Pagli, G.~Prencipe, N.~Santoro, and
  G.~Viglietta.
\newblock Gathering in dynamic rings.
\newblock {\em Theoretical Computer Science}, 2018.

\bibitem[DMGK{\etalchar{+}}06]{MDKKPV06}
G.~De~Marco, L.~Gargano, E.~Kranakis, D.~Krizanc, A.~Pelc, and U.~Vaccaro.
\newblock Asynchronous deterministic rendezvous in graphs.
\newblock {\em Theor. Comput. Sci.}, 355[3]:315--326, 2006.

\bibitem[FIP{\etalchar{+}}05]{FIPPP05}
P.~Fraigniaud, D.~Ilcinkas, G.~Peer, A.~Pelc, and D.~Peleg.
\newblock Graph exploration by a finite automaton.
\newblock {\em Theoretical Computer Science}, 345[2-3]:331--344, 2005.

\bibitem[GKM{\etalchar{+}}08]{GKMNZ07}
L.~Gasieniec, R.~Klasing, R.~Martin, A.~Navarra, and X.~Zhang.
\newblock Fast periodic graph exploration with constant memory.
\newblock {\em Journal of Computer and System Sciences}, 74[5]:808--822, 2008.

\bibitem[GTKC12]{GTKC12}
C.~Gong, S.~Tully, G.~Kantor, and H.~Choset.
\newblock Multi-agent deterministic graph mapping via robot rendezvous.
\newblock In {\em 2012 IEEE International Conference on Robotics and
  Automation}, pages 1278--1283. IEEE, 2012.

\bibitem[Ilc08]{ilcinkas2008}
D.~Ilcinkas.
\newblock Setting port numbers for fast graph exploration.
\newblock {\em Theoretical Computer Science}, 401[1-3]:236--242, 2008.

\bibitem[KLO10]{KLO10}
F.~Kuhn, N.~Lynch, and R.~Oshman.
\newblock Distributed computation in dynamic networks.
\newblock In {\em Proceedings of the 42nd ACM symposium on Theory of computing
  (STOC)}, pages 513--522. ACM, 2010.

\bibitem[PP98]{PP1998}
P.~Panaite and A.~Pelc.
\newblock Exploring unknown undirected graphs.
\newblock In {\em Proceedings of the ninth annual ACM-SIAM symposium on
  Discrete algorithms}, pages 316--322. Society for Industrial and Applied
  Mathematics, 1998.

\bibitem[SKS{\etalchar{+}}20]{SKSOKM20}
M.~Shibata, N.~Kawata, Y.~Sudo, F.~Ooshita, H.~Kakugawa, and T.~Masuzawa.
\newblock Move-optimal partial gathering of mobile agents without identifiers
  or global knowledge in asynchronous unidirectional rings.
\newblock {\em Theoretical Computer Science}, 2020.

\end{thebibliography}

\end{document}